\documentclass{llncs}

\usepackage{graphicx}
\usepackage{amsmath,amsfonts}
\usepackage{tabu,caption}

\usepackage{auto-pst-pdf}
\usepackage{pstricks,pst-node}

\usepackage{versions}
\excludeversion{sagt}  \includeversion{arxiv}

\newcommand{\citet}[1]{\cite{#1}}
\newcommand{\citep}[1]{\cite{#1}}

\newcommand{\mgft}{market-gain-from-trade}
\newcommand{\tgft}{total-gain-from-trade}

\begin{document}

\title{SBBA: a Strongly-Budget-Balanced Double-Auction Mechanism}
\titlerunning{Strongly-Budget-Balanced Double-Auction}  % abbreviated title (for running head)

\author{
Erel Segal-Halevi \and 
Avinatan Hassidim \and
Yonatan Aumann
}
%
%\authorrunning{Ivar Ekeland et al.} % abbreviated author list (for running head)
%\tocauthor{Ivar Ekeland, Roger Temam, Jeffrey Dean, David Grove, Craig Chambers, Kim B. Bruce, and Elisa Bertino}
%
\institute{
Bar-Ilan University, Ramat-Gan 5290002, Israel,\\
\email{\{erelsgl,avinatanh,yaumann\}@gmail.com}
}

\maketitle              % typeset the title of the contribution

\begin{abstract}
In a seminal paper, McAfee (1992) presented the first dominant strategy truthful mechanism for double auction.
His mechanism attains nearly optimal gain-from-trade when the market is sufficiently large. 
However, his mechanism may leave money on the table, since the price paid by the buyers may be higher than the price paid to the sellers. This money is included in the gain-from-trade and in some cases it accounts for almost all the gain-from-trade, leaving almost no gain-from-trade to the traders. 
We present SBBA: a variant of McAfee's mechanism which is strongly budget-balanced. There is a single price, all money is exchanged between buyers and sellers and no money is left on the table. This means that all gain-from-trade is enjoyed by the traders. 
We generalize this variant to spatially-distributed markets with transit costs.
\end{abstract}

\section{Introduction}
%<*sec>
In the simplest {\em double auction} a single seller has a single item. The seller values the item for $s$, which is private information to the seller. A single buyer values the item for $b$, which is private to the buyer. If $b > s$, then trade can increase the utility for both traders; there is a potential \emph{gain-from-trade} of $b-s$. However, there is no truthful, individually rational, budget-balanced mechanism that will perform the trade if-and-only-if it is beneficial to both traders. The reason is that it is impossible to determine a price truthfully. This is easy to see for a deterministic mechanism. If the mechanism chooses a price $p < b$, the seller is incentivized to bid $(p + b)/2$ to force the price up; similarly, if the mechanism chooses a price $p>s$, the buyer is incentivized to force the price down. The impossibility holds even when the valuations are drawn from a known prior distribution and even when the mechanism is allowed to randomize; see the classic papers of \citet{Myerson1983Efficient} and \citet{Vickrey1961Counter}.

\subsection{McAfee's Trade-Reduction Mechanism}
McAfee \citep{McAfee1992Dominant} showed how to circumvent this impossibility result when there are many sellers, seller $i$ having  private valuation $s_i$, and many buyers, buyer $i$ having private valuation $b_i$. In McAfee's double auction mechanism, each trader is asked to give his valuation. The sellers are sorted in an ascending order according to their valuations $s_1 \le s_2 \le \ldots \le s_n$, and the buyers are sorted in a descending order $b_1 \ge b_2 \ge \ldots \ge b_n$. Let $k$ be the largest index such that $s_k \le b_k$. The optimal gain-from-trade is attained by picking any price $p\in [s_k , b_k]$ and performing $k$ deals in that price. But this scheme is not truthful. McAfee attains truthfulness by considering the following two cases: 

(a) If there are at least $k+1$ buyers and $k+1$ sellers and the price $p_{k+1} := (b_{k+1} + s_{k+1})/2$ is in the range $[s_k , b_k]$, then $p_{k+1}$ is set as the market price, allowing all $k$ efficient deals to execute in that price.

(b) Otherwise, two prices are used, and $k-1$ deals are done: all sellers with values $s_1, \ldots, s_{k-1}$ sell their item for $s_{k}$, and all buyers with values $b_1, \ldots b_{k-1}$ buy an item for $b_k$. The mechanism performs a \emph{trade reduction} by canceling a single deal, the deal between $b_k$ and $s_k$, which is the least efficient of the $k$ efficient deals. Hence, its gain-from-trade is $1 - 1/k$ of the maximum.

Crucially, the gain-from-trade approximated by McAfee's mechanism is the \textbf{\tgft{}} - the gain-from-trade including the money left on the table due to the difference between the buyers' price and the sellers' price. Moreover, this money might include almost all the gain, so that the \textbf{\mgft{}} - the gain enjoyed by the traders - might be near zero.
\begin{example}\label{exm:intro}
There are $k$ buyers and $k$ sellers with the following valuations
\begin{itemize}
\item $s_{i}=0$ and $b_{i}=B$ \hspace{1.7cm} for all $i\in\{1,\dots,k-1\}$.
\item $s_{k}=\varepsilon$ and $b_{k}=B-\varepsilon$, \hspace{1cm} where $0<\varepsilon\ll B$.
\end{itemize}
The optimal \mgft{} occurs when all sellers sell and all buyers buy, and it is: $k\cdot B-2\varepsilon$. 

Since there are no $k+1$-th buyer and seller, McAfee's mechanism sets the buy price at $B-\varepsilon$ and the sell price at $\varepsilon$. The trade includes $k-1$ buyers and $k-1$ sellers. The gain-from-trade in each deal is $B$, so the \tgft{} is $(k-1)\cdot B$, which is a very good approximation to the optimum for large $k$.

However, the net gain of each trader is $\varepsilon$, so the \mgft{} is only $(k-1)\cdot 2\varepsilon$; when $\varepsilon\to0$, the \mgft{} becomes arbitrarily small, and most gain-from-trade $(k-1)\cdot (B-2\varepsilon)$ remains on the table.
\qed
\end{example}
Money on the table can be desirable in some cases. E.g, the government may want to arrange a double-auction between commercial firms and collect the revenue. However, in other cases it may be considered unfair and drive traders away. For example, if traders in a stock-exchange notice that most gain-from-trade is taken by the operator, they may decide to switch to another operator.

The double-auction literature, e.g. \citep{Blumrosen2014Reallocation}, differentiates between mechanisms that are \emph{weakly budget-balanced}, i.e, the auctioneer does not lose but may gain money, and \emph{strongly budget-balanced}, i.e, the auctioneer does not lose nor gain any money. McAfee's mechanism is weakly budget-balanced. 

\subsection{Our Mechanism}
In this paper we introduce SBBA - a Strongly-Budget-Balanced double-Auction mechanism. SBBA attains strong budget-balance by setting a \emph{single} trade price for all traders, in all cases. This may lead to excess supply; to handle the excess supply, a lottery is done between the sellers. At most one seller, selected at random, is excluded from trade. Hence, the expected \tgft{} of SBBA is the same as McAfee's - $1-1/k$.

A disadvantage of SBBA is that its approximation holds only in expectation (taken over the randomization of the mechanism), while McAfee's approximation holds in the worst case. An advantage of SBBA is that it is strongly budget-balanced, so the \mgft{} equals the \tgft{} - all gain-from-trade is enjoyed by the traders. Besides these differences, SBBA has all the desirable properties of McAfee's mechanism:
\begin{itemize}
\item It is \textbf{ex-post individually-rational} - a trader never loses any value from participating in the market (a buyer is never forced to buy an item for more than its declared value; a seller is never forced to sell an item for less than its declared value; a trader who does not participate in the trade pays nothing).
\item It is \textbf{ex-post dominant-strategy truthful}: for every trader, every vector of declarations by the other traders and every randomization, the trader's net value is always maximized by reporting his true value.
\item It is \textbf{prior-free} - it does not assume or require any knowledge on the distribution of the traders' valuations. In other words, its approximation ratio is valid even for adversarial (worst-case) valuations.
\end{itemize}
Below we survey some related literature (Section \ref{sec:related}). Then, we present the SBBA mechanism in the most basic double auction setting - a single market with single-unit buyers and single-unit sellers (Section \ref{sec:sbba}). The idea of SBBA can be used in much more complex settings. To demonstrate its generality, we show how to use it in a \emph{spatially-distributed market} - a collection of markets in different locations, with positive transit costs between markets (Section \ref{sec:sbba-sdm}).

\section{Related Work}
\label{sec:related}
\subsection{\textbf{Double auctions}}
VCG (Vickrey-Clarke-Grove) is a well-known mechanism that can be used in various settings, including double auction. It is truthful and attains the maximum gain-from-trade. Its main drawback is that it has budget deficit, which means that the auctioneer has to subsidize the market.

McAfee's Trade-Reduction mechanism was extended and generalized in many ways. Some of the extensions are surveyed below.

Babaioff et al extend McAfee's mechanism to handle spatially-distributed markets with transit costs \citep{Babaioff2009Mechanisms} and supply chains \citep{Babaioff2004Concurrent,Babaioff2005Incentivecompatible,Babaioff2006Incentive}, providing similar welfare guarantees. One variant, \textbf{Probabilistic Reduction} \citep{Babaioff2004Concurrent}, achieves \emph{ex-ante} budget balance by randomly selecting between the Vickrey-Clarke-Grove (VCG) mechanism and the Trade-Reduction mechanism. The probability is selected such that the deficit of the VCG exactly balances (in expectation) the surplus of McAfee. However, the probability depends on the distribution of the agents' valuations so the mechanism is not prior-free.

Lately, there has been a surge of interest in a more complicated market, namely double spectrum auctions, in which an auction is used to transfer spectrum from incumbent companies (e.g. TV stations) to modern companies (e.g. cellular operators).
\citet{Yao2011Efficient} adapted the Trade-Reduction mechanism to a double-spectrum-auction by creating groups of non-interfering buyers that can buy the same channel. \citet{Wang2010TODA} created on online variant of Trade-Reduction to handle the case in which new buyers arrive over time. \citet{Wang2011District}
adapted Trade-Reduction to enable local markets, in which only some buyer-seller combinations are feasible. \citet{Feng2012TAHES} adapted
it to enable heterogeneous spectra. All these mechanisms are only weakly budget-balanced.

Some recent papers extend McAfee's mechanism to settings with more than one item per trader. These are the \textbf{TAHES} mechanism of \citet{Feng2012TAHES}, which is multi-type
single-unit; the \textbf{Secondary Market} mechanism of \citet{Xu2010Secondary}, which is single-type multi-unit; and the \textbf{Combinatorial Reallocation} mechanism of  \citet{Blumrosen2014Reallocation}, which is multi-type multi-unit.  Our mechanism is single-type single-unit; we leave  to future work its extension to multi-type and  multi-unit settings.

A different approach to double auctions is \emph{random-sampling}. It was introduced by  \citep{Baliga2003Market} under the assumptions that the traders' valuations are random variables drawn from an unknown bounded-support distribution and there is a single item-type.
We recently extended it to a prior-free setting with multiple item-types \citep{SegalHalevi2016RandomSampling}. The idea is to divide the traders randomly to two half-markets, calculate an optimal price in one half and apply it to the other half, and vice versa. Since there is a single price in each market, the mechanism is strongly-budget-balanced. However, our analysis (for the single-type case) shows that the gain-from-trade is $1-O(\sqrt{\ln{k}/k})$. While this still approaches 1 when the market is sufficiently large, the convergence rate is much slower than the $1-1/k$ guarantee of SBBA.

Recently, \citet{ColiniBaldeschi2016Approximately} presented a \textbf{two-sided sequential posted price mechanism (2SPM)}. Its main objective is to handle matroid constraints on the sets of buyers that can be served simultaneously. Like our mechanism, it is strongly budget-balanced. However, its approximation ratio is multiplicative - the ratio is between 4 to 16, depending on the setting. In particular, the approximation ratio does not approach 1 when the market is large. 

The following table compares our work to some typical single-type single-unit double-auction mechanisms. In this table, an asterisk means "in expectation". TGFT means \tgft{} and MGFT means \mgft{}; they are identical for strongly-budget-balanced mechanisms.
\begin{center}
\begin{tabular}{|c|c|c|c|c|}
	\hline
	\textbf{Mechanism} & \textbf{Prior-free} & \textbf{Budget} & \textbf{TGFT} & \textbf{MGFT} \\
	\hline
	\hline
	VCG & Yes & Deficit & 1 & 1 \\
	\hline
	Trade Reduction (McAfee 1992)\cite{McAfee1992Dominant} & Yes & Surplus & $1-1/k$ & 0 \\
	\hline
	Random Sampling (Baliga 2003)\cite{Baliga2003Market}\cite{SegalHalevi2016RandomSampling} & Yes & Balance &  \multicolumn{2}{c|}{$1 - O(\sqrt{\ln{k}/k})$} \\
	\hline
	Probabilistic Reduction (Babaioff 2009)\cite{Babaioff2009Mechanisms} & No & Balance{*} &  \multicolumn{2}{c|}{$1-1/k${*}}  \\
	\hline
	2SPM (Colini 2016)\cite{ColiniBaldeschi2016Approximately} & No & Balance &  \multicolumn{2}{c|}{1/4 to 1/16} \\
	\hline
	\hline
	SBBA (this paper) & Yes & Balance &  \multicolumn{2}{c|}{$1-1/k${*}} \\
	\hline
\end{tabular}
\end{center}
%{\scriptsize Table notes: An asterisk ({*}) means ``in expectation''. In non-prior-free mechanisms, the expectation is taken over the assumed prior. In our mechanism, the expectation is over the randomization of the mechanism. TGFT means \tgft{} and MGFT means \mgft{}; they may be different if the mechanism is not strongly-budget-balanced. The numbers are approximation factors relative to the optimal gain. Higher is better: 1 means optimal, 0 means that all welfare might be lost. $k$ is the optimal number of deals.}

\subsection{\textbf{Redistribution mechanisms}}
The problem illustrated by Example \ref{exm:intro}, where the money left on the table eats most of the welfare and leaves little welfare to the agents, happens in other domains besides double auctions. Several authors have suggested a two-step solution: in the first step, the original mechanism is executed and the budget-surplus is collected. In the second step, some of the surplus is re-distributed among the agents. This second step is called a \emph{redistribution mechanism} and it should be carefully designed in order not to harm the truthfulness of the original mechanism. While it is not possible to redistribute the entire surplus in a truthful way, there are some truthful mechanisms that redistribute a large fraction of the surplus
%\cite{Cavallo2006Optimal,Guo2008Optimalinexpectation,Guo2008Undominated,Apt2008Welfare}. 
\cite{Cavallo2006Optimal,Guo2008Optimalinexpectation,Apt2008Welfare}. 
 We take a different approach: we modify the original mechanism such that there is no budget-surplus at all, so no redistribution is needed and all social welfare remains with the agents.

\section{The SBBA Mechanism}
\label{sec:sbba}
Order the buyers and sellers as in McAfee's mechanism. In case of ties, impose an arbitrary order, e.g. lexicographic order of name.

Let $k$ be the largest integer for which $s_k\leq b_k$ (or 0
if already $s_1>b_1$). Call the first $k$ sellers, the ``cheap
sellers'', and the first $k$ buyers, the ``expensive buyers''.
As a convenience, if $k=n$, set $s_{k+1}=\infty$ and $b_{k+1}=0$. 
Note that with this notation, we have that $s_i\leq b_i$ for
$i\le k$ and $s_{i}>b_{i}$ for $i>k$. The price is:
\begin{align*}
p := \min(s_{k+1},b_k) \enspace .
\end{align*}
There are two cases. We illustrate them below by plotting buyers' valuations as balls and sellers' valuations as squares (in all illustrations, $k=3$).

\newcommand{\buyertitle}[1]{\rput(-1,#1){\color{blue}B:}}
\newcommand{\buyer}[2]{\rput(#1,#2){\pscircle[fillstyle=solid,fillcolor=blue,linecolor=blue](0,0){0.2}}} 

\newcommand{\sellertitle}[1]{\rput(-1,#1){\color{green}S:}}
\newcommand{\seller}[2]{\rput(#1,#2){\rput(-0.2,-0.2){\psframe[fillstyle=solid,fillcolor=green,linecolor=green](0.4,0.4)}}}

\newcommand{\priceline}[1]{\psline[linestyle=solid,linecolor=black!50](0,#1)(10,#1)}

\psset{unit=3mm}
\newcommand{\buyerx}[2]{\buyer{#1}{#2}\priceline{#2}}
\newcommand{\sellerx}[2]{\seller{#1}{#2}\priceline{#2}}

\psset{unit=3mm}

\noindent \textbf{Case 1}: $\text{s}_{k+1}\leq b_{k}$ (note that
by definition of $k$: $s_{k+1}>b_{k+1}$):

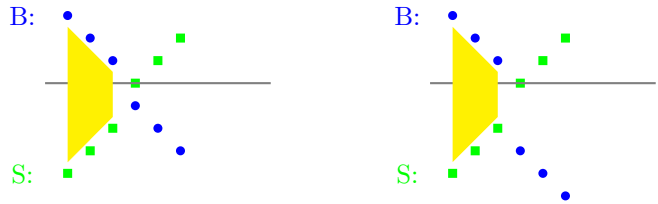
\begin{center}
\begin{pspicture}(10,10)
\buyertitle{8}
\buyer{1}{8}
\buyer{2}{7}
\buyer{3}{6}
\buyer{4}{4}
\buyer{5}{3}
\buyer{6}{2}
\sellertitle{1}
\seller{1}{1}
\seller{2}{2}
\seller{3}{3}
\sellerx{4}{5}
\seller{5}{6}
\seller{6}{7}
\pspolygon[fillstyle=solid,fillcolor=yellow,linestyle=none](1,7.5)(2,6.5)(3,5.5)(3,3.5)(2,2.5)(1,1.5)
\end{pspicture}
\hskip 2cm
\begin{pspicture}(10,10)
\buyertitle{8}
\buyer{1}{8}
\buyer{2}{7}
\buyer{3}{6}
\buyer{4}{2}
\buyer{5}{1}
\buyer{6}{0}
\sellertitle{1}
\seller{1}{1}
\seller{2}{2}
\seller{3}{3}
\sellerx{4}{5}
\seller{5}{6}
\seller{6}{7}
\pspolygon[fillstyle=solid,fillcolor=yellow,linestyle=none](1,7.5)(2,6.5)(3,5.5)(3,3.5)(2,2.5)(1,1.5)
\end{pspicture}
\end{center}
The price is $p = s_{k+1}$. All $k$ expensive buyers and $k$ cheap sellers trade in $p$. 

\noindent \textbf{Case }2: $\text{s}_{k+1}>b_{k}$ :
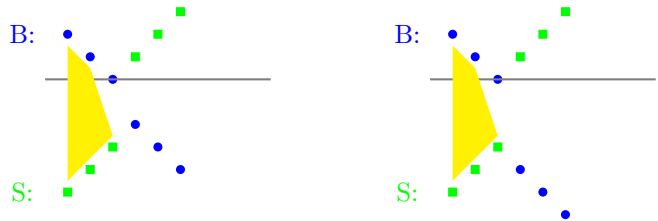
\begin{center}
\begin{pspicture}(10,10)
\buyertitle{8}
\buyer{1}{8}
\buyer{2}{7}
\buyerx{3}{6}
\buyer{4}{4}
\buyer{5}{3}
\buyer{6}{2}
\sellertitle{1}
\seller{1}{1}
\seller{2}{2}
\seller{3}{3}
\seller{4}{7}
\seller{5}{8}
\seller{6}{9}
\pspolygon[fillstyle=solid,fillcolor=yellow,linestyle=none](1,7.5)(2,6.5)(3,3.5)(2,2.5)(1,1.5)
\end{pspicture}
\hskip 2cm
\begin{pspicture}(10,10)
\buyertitle{8}
\buyer{1}{8}
\buyer{2}{7}
\buyerx{3}{6}
\buyer{4}{2}
\buyer{5}{1}
\buyer{6}{0}
\sellertitle{1}
\seller{1}{1}
\seller{2}{2}
\seller{3}{3}
\seller{4}{7}
\seller{5}{8}
\seller{6}{9}
\pspolygon[fillstyle=solid,fillcolor=yellow,linestyle=none](1,7.5)(2,6.5)(3,3.5)(2,2.5)(1,1.5)
\end{pspicture}
\end{center}
The price is $p = b_{k}$. From the group of $k$ cheap sellers, select $k-1$ at random and let them trade with the $k-1$ expensive buyers (excluding $b_k$). 

\begin{theorem}
\label{thm:sbba}
The SBBA mechanism is prior-free (PF), individually-rational (IR), strongly-budget-balanced (SBB) and dominant-strategy truthful. Its expected \mgft{} is at least $1-1/k$ of the optimum.
\end{theorem}
\begin{proof}
The mechanism is PF by construction. It is IR since the trade is always between buyers whose value is above $p$ and sellers whose value is below $p$. It is SBB since there is always a single price and all payments are between buyers and sellers.
To analyze the gain-from-trade, note that in case 1, all $k$ efficient deals are carried out and thus the maximum possible gain is achieved. In case 2, a single random deal is canceled, which implies an expected loss of $1/k$. Hence, in expectation, at most a fraction $1/k$ of the gain-from-trade is lost.

To prove truthfulness, we use a characterization of truthful single-parameter mechanisms from
\cite[chapter 9]{Nisan2007Introduction}. A mechanism is truthful iff:

(a) the probability of an agent to win, given the bids of other agents, is a weakly monotonically increasing function of the agent's bid, and:

(b) the price paid by a winning agent equals the \emph{critical price} - the lowest value this agent has to bid in order to win, given the other agents' bids.

%Truthfulness is proved by the following two lemmas.

We prove that SBBA is truthful for the \textbf{buyers}. \hspace{1cm} The winning probability of a buyer is either 0 or 1. Hence, it is sufficient to prove that a winning buyer never loses by raising the bid. Consider two cases of a winning buyer:
\begin{itemize}
\item The buyer is $b_i$ for $i<k$. 
\item The buyer is $b_k$ and $b_k\geq s_{k+1}$.
\end{itemize}
In both cases, if the bid is raised, $b_k$ remains above $s_{k+1}$ and the buyer's index may only decrease, so the buyer still wins.

The critical price when $b_k<s_{k+1}$ is $b_k$, since a trading buyer (one of the expensive $k-1$ buyers) exits the trade by bidding below $b_k$ and becoming the new $b_k$. 
The critical price when $b_k\geq s_{k+1}$ is $s_{k+1}$, since a trading buyer (one of the expensive $k$ buyers) exits the trade by becoming $b_k$ and bidding below $s_{k+1}$. In both cases, the price paid by the winning buyers is the critical price. 

Finally, we prove that SBBA is truthful for the \textbf{sellers}. \hspace{1cm} The sellers' ask-prices represent negative valuations, so monotonicity means that a seller's probability of participation should increase when the ask-price \emph{decreases}. The winning probability of a seller is 0 when the seller is $s_i$ for $i\geq k+1$. When the seller is $s_i$ for $i\leq k$, the winning probability is positive, and it does not depend on $s_i$ itself but only on the relation between $s_{k+1}$ and $b_k$:
\begin{itemize}
\item If $s_{k+1}\leq b_k$, the probability that $s_i$ (for $i\leq k$) wins and pays is $1$;
\item If $s_{k+1}> b_k$, the probability that $s_i$ (for $i\leq k$) wins and pays is $1-1/k$.
\end{itemize}
In each of these cases, decreasing the ask-price can only decrease the seller's index, so the winning probability remains the same.

The critical price when $s_{k+1}\leq b_k$ is $s_{k+1}$, since a trading seller (one of the $k$ cheap sellers) exits the trade by asking above $s_{k+1}$ and becoming the new $s_{k+1}$. This is indeed the price paid to a winning seller.

The critical price when $s_{k+1} > b_k$ is $b_k$, since a trading seller (one of the $k$ cheap sellers) exits the trade by asking above $b_k$, which decreases $k$ by 1 (the seller who increased his ask-price becomes $s_k$, but now the number of efficient deals is $k-1$, so $s_k$ is excluded from trade). Indeed, $b_k$ is the price paid to a winning seller. Note that in this case both the winning and the price are realized with probability $1-1/k$. 

\qed
\end{proof}

\subsection{Alternatives}
\label{sub:Economic-interpretation}
The price set by our mechanism, $\min(b_{k},s_{k+1})$, has an interesting economic interpretation: it is the highest price in a price-equilibrium (aka Walrasian equilibrium), i.e. the highest price in which the market can be cleared by balancing supply and demand. If the price is raised above $b_k$ then the demand becomes less than $k$ while the supply is still at least $k$; if the price is raised above $s_{k+1}$ then the supply becomes at least $k+1$ while the demand is still at most $k$; in both cases there is an excess supply.

It is possible to switch the role of buyers and sellers, splitting the cases by whether $b_{k+1}>s_k$. In this case, the price is $\max(s_{k},b_{k+1})$ which is the \emph{lowest} price-equilibrium. The alternative mechanism has the same properties of our original mechanism.

There are some other alternatives that come to mind, but are either not truthful or not efficient:
\begin{itemize}
	\item If in Case 2, instead of using a lottery we select the trading sellers deterministically (e.g. taking the $k-1$ cheaper sellers), the mechanism will not be truthful, since some agents who want to trade at the market price have an incentive to deviate from their true values in order to enter the trade.
	\item If we always set the price to $s_{k+1}$, in some cases this price might be higher than the valuations of all buyers so we might lose all gain-from-trade. 
	\item If we always set the price to $b_{k}$, in some cases this price might be higher than the valuations of all sellers, the market will be flooded by inefficient sellers, and the expected gain-from-trade will be low.
\end{itemize}

\section{Spatially Distributed Markets}
\label{sec:sbba-sdm}
A \emph{spatially distributed market} is a collection of several markets, each of which is located in a different geographic location. It is possible to transport goods from one market to another for a fixed, positive transit cost, which may be different for each ordered pair of markets. Babaioff et al \citet{Babaioff2009Mechanisms} extended McAfee's mechanism to handle such markets. Similarly to McAfee's mechanism, their mechanism has a budget surplus. Below we briefly present their mechanism and present a strongly budget-balanced variant of it.

\subsection{Create the market-flow graph}
Create a network-flow graph representing the market in the following way:
\begin{itemize}
\item Create a node for each market. Create a directed edge between each pair of markets, with infinite capacity and with cost equal to the transit cost between the two markets (which may be different in each direction).
\item Create an additional \emph{Agents node}, representing the buyers and sellers.
%\footnote{It is called a "sink node" in \citet{Babaioff2009Mechanisms}, but it is not really a sink since it has outgoing edges.} 
For each seller in market $i$, create an edge FROM the Agents node TO Market $i$, with unit capacity and cost equal to the seller's ask-price. This edge represents the seller producing an item and sending it to the market. For each buyer in market $i$, create an edge TO the Agents node FROM Market $i$, with unit capacity and cost equal to MINUS the buyer's bid. This edge represents the buyer bringing an item from the market.
\end{itemize}
The following illustration shows a graph representing two markets, with a transit cost of 4 in each direction. Each solid arc represents an infinite-capacity edge. Each dashed arc represents several unit-capacity edges with different costs. The numeric labels are the costs.

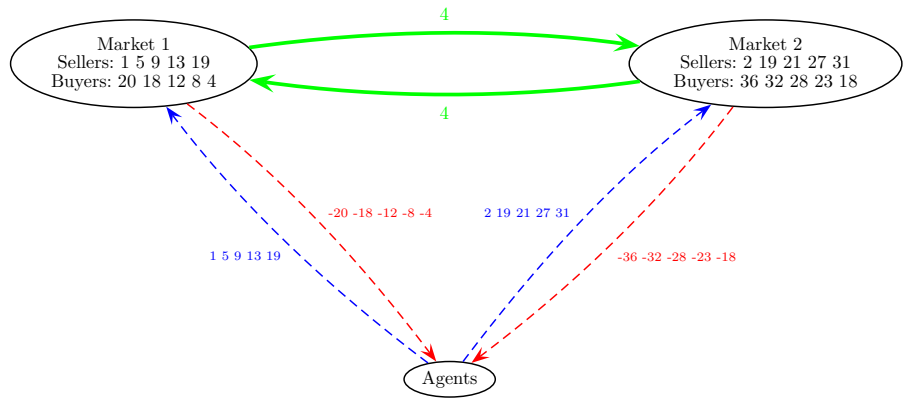
\psset{unit=1.5cm,arrowsize=7pt}
\def\market(#1,#2,#3,#4,#5){%
	\rput(#2,#3){%
		\ovalnode{M#1}{%
			\shortstack{Market #1\\Sellers: #4\\Buyers: #5}%
		}%
	}%
}
\def\Agents{
\rput(0,0){\ovalnode[linecolor=black]{Agents}{\color{black}Agents}}
}
\def\sellers(#1,#2){
	\ncarc[npos=0.7,linecolor=blue,linestyle=dashed]{->}{Agents}{M#1}\naput{\color{blue}\scriptsize #2}
}
\def\buyers(#1,#2){
	\ncarc[npos=0.3,linecolor=red,linestyle=dashed]{->}{M#1}{Agents}\naput{\color{red}\scriptsize #2}
}
\def\transit(#1,#2,#3) {
	\ncarc[linewidth=2pt,linecolor=green,arrowsize=10pt]{->}{#1}{#2}\naput{\color{green}#3}
}

\def\exampleA{  % the buyer in market 1 sets the price
\market(1,-4,4,1 5 9 13 17,20 16 12 8 4)
\sellers(1,1 5 9 13 17)
\buyers(1,-20 -16 -12 -8 -4)
\market(2,4,4,15 19 22 27 31,36 32 28 23 18)
\sellers(2,15 19 22 27 31)
\buyers(2,-36 -32 -28 -23 -18)
\transit(M1,M2,4)
\transit(M2,M1,4)
}
\def\exampleB{  % the buyer in market 2 sets the price
\market(1,-4,4,1 5 9 13 19,20 18 12 8 4)
\sellers(1,1 5 9 13 19)
\buyers(1,-20 -18 -12 -8 -4)
\market(2,4,4,2 19 21 27 31,36 32 28 23 18)
\sellers(2,2 19 21 27 31)
\buyers(2,-36 -32 -28 -23 -18)
\transit(M1,M2,4)
\transit(M2,M1,4)
}
\begin{center}
\scalebox{0.7}{
\begin{pspicture}(-6,-1)(6,4)
\Agents
\exampleB
\end{pspicture}
}
\end{center}
%A deal between a seller and a buyer is represented by a cycle in the market-graph: starting at the Agents node, follow the seller's edge to the seller's market, then follow the transit edges to the buyer's market, and finally follow the buyer's edge back to the Agents node. If the total cost of the cycle is negative, then the deal has a positive gain-from-trade since then the value of the buyer is larger than the total cost of production plus transit.

\subsection{Calculate a minimum-cost flow}
Babaioff et al use a known polynomial-time algorithm for finding a flow with minimum cost in the market graph. Assuming all data is integral, the flow in every edge is also an integer number. In particular, the flow in every buyer/seller edge (with capacity 1) is either 0 or 1. Hence, a flow in the graph defines an allocation in which each trader trades if and only if the flow in the corresponding edge to/from the Agents node is 1. A minimum-cost flow corresponds to an optimal trade: the (negative) cost of the flow is minus the gain-from-trade.  The following illustration shows the minimum-cost flow and the corresponding optimal trade in the above example market (non-trading agents are bracketed):
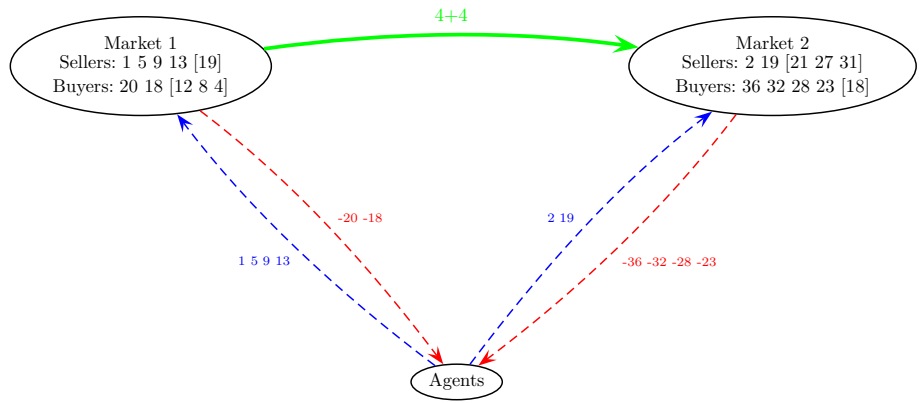
\def\exampleAA{  % the buyer in market 1 sets the price
\market(1,-4,4,1 5 9 13 [17],20 16 [12 8 4])
\sellers(1,1 5 9 13)
\buyers(1,-20 -16)
\market(2,4,4,15 19 [22 27 31],36 32 28 23 [18])
\sellers(2,15 19)
\buyers(2,-36 -32 -28 -23)
\transit(M1,M2,4+4)
}
\def\exampleBB{  % the buyer in market 2 sets the price
\market(1,-4,4,1 5 9 13 [19],20 18 [12 8 4])
\sellers(1,1 5 9 13)
\buyers(1,-20 -18)
\market(2,4,4,2 19 [21 27 31],36 32 28 23 [18])
\sellers(2,2 19)
\buyers(2,-36 -32 -28 -23)
\transit(M1,M2,4+4)
}
\begin{center}
\scalebox{0.7}{
\begin{pspicture}(-6,-1)(6,4)
\Agents
\exampleBB
\end{pspicture}
}
\end{center}
There are 6 efficient deals. The net cost is -100 so the gain-from-trade is 100.

\subsection{Find commercial-relationship components}
In the optimal flow, the markets can be partitioned to groups, such that all trade is within groups and no trade is between groups. Such groups are called \emph{commercial-relationship components}; they are the connectivity components of a graph in which the nodes are the markets and there is an edge between markets trading in the optimal flow.

The optimal trade can be attained in a price-equilibrium, in which there is a single price in each market. In each component, the prices in the different markets are tied by the equilibrium conditions: if the price in market $i$ is $p_i$, and there is positive trade from market $i$ to market $j$, then we must have:
\begin{align}\label{eq:eq}
p_j = p_i+Cost[i,j]
\end{align}
\newcommand{\deltaij}{\ensuremath{\Delta_{i,j}}}
since in equilibrium, the sellers in market $i$ should be indifferent between selling in their local market for a net revenue of $p_i$, and selling in market $j$ for a net revenue of $p_j-Cost[i,j]$. Therefore, in each component, setting the price in a single market uniquely determines the prices in other markets. Formally, for every two markets $i,j$ in the same component there is a constant \deltaij{} such that, in any price-equilibrium, $p_j = p_i + \deltaij$ (\deltaij{} can be calculated, for example, by calculating cheapest paths in the residual graph of the min-cost flow; see \cite{Babaioff2009Mechanisms}).

In our running example, there is a single component. Here, $\Delta_{1,2}=Cost[1,2]=4$. The optimal trade can be attained in a price-equilibrium in which the price in Market 1 is $p_1$ and the price in Market 2 is $p_2=p_1+4$. Any price-vector between $p_1=15,p_2=19$ and $p_1=17,p_2=21$ is an equilibrium price-vector.

\subsection{Trade reduction}
At this point, Babaioff et al calculate a \emph{reduced residual graph} of the min-cost flow, remove a single cycle representing the least efficient deal, and determine two prices in each market (buy price and sell price) based on distances in the reduced residual graph. This gives a truthful mechanism with a budget surplus.

Here our mechanism takes a different approach:
\begin{itemize}
\item In each commercial-relationship-component, virtually bring all traders to an arbitrary market in that component, e.g, Market $i$. Adjust their bids according to the price-equilibrium conditions: a bid $b$ in Market $j$ is translated to a bid $b-\deltaij$ in Market $i$.
\item Proceed as in the single-market situation (Section \ref{sec:sbba}): order the buyers decreasingly and the sellers increasingly and find $k$ - the total number of efficient deals in the component. Set the price in Market $i$ to $p_i := \min(b_k,s_{k+1})$. 
\item Determine the prices in the other markets of the same component according to the equilibrium conditions: for every market $j$, set $p_j:= p_i + \deltaij$. Below we prove that all these prices are non-negative.
\item If $b_k$ is the price-setter, then exclude $b_k$ and a random seller in the component from trading. Otherwise, allow all $k$ efficient traders in the component to trade in their market prices.
\item Ban any trade between different components.
\end{itemize}

In the above example, when all traders are brought to Market 1, we have the following valuations (here $\Delta_{1,2}=4$; the adjusted valuations of traders brought from Market 2 are displayed in slanted digits):
\begin{itemize}
\item Sellers: \emph{$-$2}, 01, 05, 09, 13, \emph{15, 17}, 19, \emph{23, 27}
\item Buyers: \emph{32, 28, 24}, 20, \emph{19}, 18, \emph{14}, 12, 08, 04
\end{itemize}
Here, the number of efficient deals $k=6$. We have $b_k=18$ and $s_{k+1}=17$, so the price-setter is $s_{k+1}$, who is originally a seller in Market 2 (where his valuation is 21). The price-vector is $p_1=17,p_2=21$. All 6 efficient deals are performed.  
\begin{arxiv}
Another example, where the price-setter is $b_k$, is shown in Appendix \ref{sec:example}.
\end{arxiv}

\subsection{Analysis}
Similarly to our single-market mechanism, the spatially-distributed-market mechanism is prior-free and strongly-budget-balanced, since there is a single price in each market. The reduction in trade is only at most a single deal \emph{in each component}. Hence, the expected gain-from-trade in each component is $(1-1/k)$ of the optimum in that component, where $k$ is the number of efficient deals in the component. When the components are large, this $k$ may be much larger than the number of efficient deals in each market alone. 

\noindent For truthfulness, we use the monotonicity characterization shown in Theorem \ref{thm:sbba}.

First, we have to prove that a winning trader never loses by increasing/decreasing the bid/ask price. Indeed, if a trader is winning, then the edge from the Agents node to the trader is active in the min-cost flow. When the bid/ask increases/decreases, the cost of that edge decreases. Hence, the cost of the flow decreases, and it remains the min-cost flow. Thus, the partition of the graph to commercial relationship components does not change. Within each component, increasing/decreasing the bid/ask price weakly decreases the index of the trader in the ordering, so a winning buyer/seller is still among the first $k$ buyers/sellers.

Next, we have to prove that each trader pays the critical price. For traders originally from Market $i$, the critical price is $p_i = \min(b_k,s_{k+1})$; this follows immediately from the proof of Theorem \ref{thm:sbba}. Consider now a winning buyer from Market $j\neq i$ in the same component as Market $i$. If our buyer bids $b$, then the bid is translated to Market $i$ as $b-\deltaij$. The buyer exits the trade when $b-\deltaij<p_i$. Hence, the critical price for our buyer is $p_i+\deltaij$, which is exactly the price $p_j$ paid by our buyer. Similar considerations are true for the sellers. 

Finally, as promised, we prove that the prices determined by our mechanism are non-negative. Let Market $j$ be a market with a smallest price in some component. Suppose by contradiction that $p_j<0$. Since we proved that the mechanism is truthful, it implies that there are no active sellers in Market $j$ (since a seller would prefer to lie than to sell in negative price). This means that there is no trade outgoing from Market $j$. Since there must be active sellers elsewhere in the component, there must be other markets in the component, and this means that there is trade incoming to Market $j$, say, from Market $i$. But this means that $p_i = p_j - Cost[i,j]$. Since we assume that transit costs are positive, this contradicts the minimality of $p_j$.
\qed
\section{Future Work}
1. Our strongly-budget-balanced mechanism is randomized. This is in contrast to the VCG and McAfee mechanisms, which are deterministic. This raises the following open question: is there a deterministic mechanism that is strongly-budget-balanced (in addition to being prior-free, truthful, individually-rational and approximately-optimal)? 

2. Besides the spatially-distributed markets described in this paper, there are many other variants of McAfee's mechanism. An interesting line of future work is to survey these variants and see if and how they can be made strongly-budget-balanced.

\paragraph{\textbf{Acknowledgments}.}
This research was funded in part by the following institutions: The Doctoral Fellowships of Excellence Program at Bar-Ilan University, the Mordechai and Monique Katz Graduate Fellowship Program, and the Israel Science Fund grant 1083/13. We are grateful to Moshe Babaioff for his advice on the spatially-distributed-market mechanism.

\begin{arxiv}
	\newpage
\appendix
\section{Another Example of a Spatially-Distributed Market}\label{sec:example}
Consider a spatially-distributed market represented by the following graph:
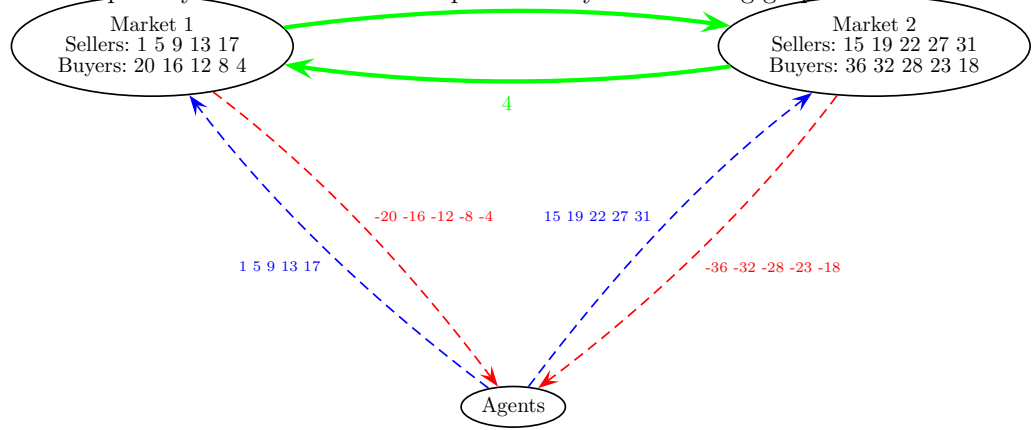
\begin{center}
	\scalebox{0.8}{
		\begin{pspicture}(-6,-1)(6,4)
		\Agents
		\exampleA
		\end{pspicture}
	}
\end{center}

Here is the optimal flow:
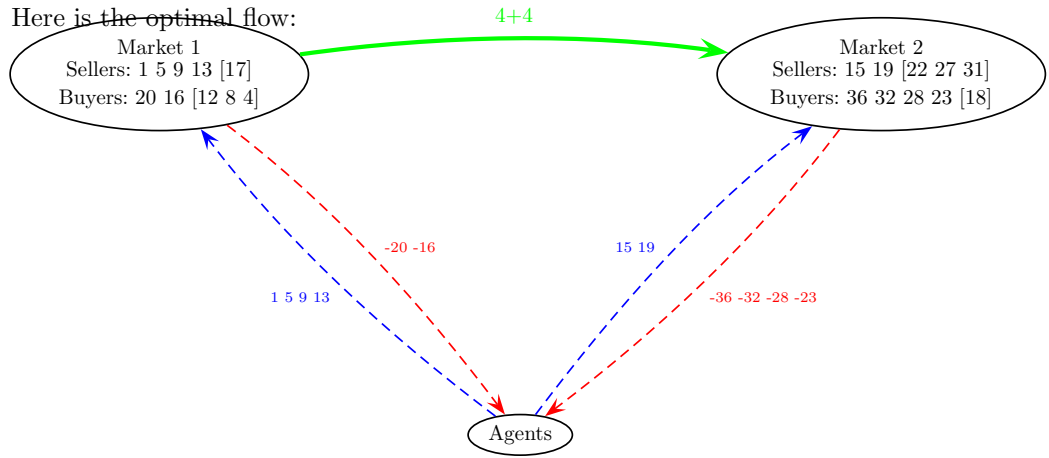
\begin{center}
	\scalebox{0.8}{
		\begin{pspicture}(-6,-1)(6,4)
		\Agents
		\exampleAA
		\end{pspicture}
	}
\end{center}
It can be attained in a price-equilibrium where the price-vector can be anywhere between $p_1=15,p_2=19$ and $p_1=16,p_2=20$ (with $p_2=p_1+4$).

When all traders are brought to Market 1, we get:
\begin{itemize}
	\item Sellers: 01, 05, 09, \emph{11}, 13, \emph{15}, 17, \emph{18, 23, 27}
	\item Buyers: \emph{32, 28, 24}, 20, \emph{19}, 16, \emph{14}, 12, 08, 04
\end{itemize}
The number of efficient deals $k=6$. $b_k=16$ and $s_{k+1}=17$, so the price-setter is $b_k$ (who is originally from Market 1). The price-vector is set to $p_1=16,p_2=20$. Buyer 16 from market 1 and one randomly-selected seller are excluded from trading. All in all, 5 out of 6 efficient deals are performed.
\end{arxiv}

\bibliographystyle{splncs03}
%\bibliography{../erelsegal-halevi}

\end{document}